\documentclass[pdflatex,a4paper]{article}
\usepackage{graphicx}
\usepackage{hyperref}
\usepackage{multirow,amsmath,amssymb,amsfonts,amsthm}
\theoremstyle{plain}
\newtheorem{theorem}{Theorem}[subsection]
\newtheorem{proposition}{Proposition}[subsection]
\newtheorem{lemma}{Lemma}[subsection]

\theoremstyle{remark}

\theoremstyle{definition}
\newtheorem{definition}{Definition}[subsection]
\author{{Nathan Thomas} {Carruth}\thanks{nathan.carruth@cantab.net}}
\title{Trapped surface formation via radial boosts}

\newcommand{\chihat}{\hat{\chi}}
\newcommand{\chihatz}{\chihat_0}
\newcommand{\R}{\mathbb{R}}
\newcommand{\Sph}{\mathbb{S}}
\newcommand{\chrangulvar}{\omega}
\newcommand{\chrubar}{{\underline{u}}}
\newcommand{\curangulvar}{\omega}
\newcommand{\curangulvarp}{{\curangulvar'}}
\newcommand{\curdr}{{\delta r}}
\newcommand{\curdrp}{{\curdr'}}
\newcommand{\curphi}{\phi}
\newcommand{\curphip}{{\curphi'}}
\newcommand{\currz}{{r_0}}
\newcommand{\curtheta}{\theta}
\newcommand{\curthetap}{{\curtheta'}}
\newcommand{\curt}{t}
\newcommand{\curtp}{{\curt'}}
\newcommand{\curubar}{{\underline{u}}}
\newcommand{\curubarp}{{\curubar'}}
\newcommand{\curubarrot}{{\curubar''}}
\newcommand{\curu}{u}
\newcommand{\curup}{{\curu'}}
\newcommand{\cururot}{{\curu''}}
\newcommand{\mink}{\eta}
\newcommand{\etaib}{{\mink_\infty}}
\newcommand{\etap}{{\mink'}}
\newcommand{\etarot}{{\mink''}}
\newcommand{\etazp}{\etaib}
\newcommand{\krangulvar}{\omega}
\newcommand{\krubar}{{\underline{u}}}
\newcommand{\ntcubar}{{\underline{u}}}
\newcommand{\ntcu}{u}
\newcommand{\ntcx}{x}
\newcommand{\rotangle}{\psi}
\newcommand{\schi}{\chi}
\newcommand{\tr}[1]{\text{tr} #1}
\renewcommand{\L}{L}

\newcommand{\Lbar}{{\underline{L}}}

\newcommand{\Lbars}{{\Lbar}'}

\newcommand{\Lib}{L}
\newcommand{\Libbar}{{\underline{L}}}

\newcommand{\Ls}{{L'}}

\newcommand{\Osize}{{\cal O}}

\newcommand{\Riem}{{\bf R}}

\newcommand{\Ric}{{\rm Ric}}

\newcommand{\chibar}{{\underline\chi}}

\newcommand{\chihati}{\hat{\chi}{}^0}

\newcommand{\chihatsi}{\hat{\chi}{}'{}^0}
\newcommand{\chisi}{\chi'{}^0}

\newcommand{\conffact}{\Omega}

\newcommand{\cscal}{\delta} 

\newcommand{\up}[1]{\overline{#1}}
\newcommand{\cscalup}{\up{\cscal}}

\newcommand{\scale}[1]{{#1'}}
\newcommand{\curangulvars}{\scale{\curangulvar}}

\newcommand{\final}[1]{#1_0}
\newcommand{\curubarf}{\final{\curubar}}

\newcommand{\curubars}{{\scale{\curubar}}}

\newcommand{\curubarsff}{\curubars_1} 

\newcommand{\curuf}{\final{\curu}}

\newcommand{\curus}{{\scale{\curu}}}

\newcommand{\curusff}{\curus_1} 

\newcommand{\datvar}{{\bf x}}

\newcommand{\datvari}{\datvar^0}

\newcommand{\datvarii}{{\datvari_0}}

\newcommand{\ddatvar}{\delta\datvar_0}

\newcommand{\dscal}{\epsilon} 

\newcommand{\dscalup}{\up{\dscal}}

\newcommand{\eframe}{e}

\newcommand{\eframei}{e^0}

\newcommand{\framea}{A}

\newcommand{\frameb}{B}

\newcommand{\framec}{C}

\newcommand{\framed}{D}

\newcommand{\frameinda}{\framea}

\newcommand{\frameindb}{\frameb}

\newcommand{\freg}{\Omega}

\newcommand{\init}{H}

\newcommand{\initin}{\init_0}

\newcommand{\initmet}{h}

\newcommand{\initout}{\init_1}

\newcommand{\interval}{(-a, a)}
\newcommand{\intervalp}{(-a', a')}

\newcommand{\metd}{t}

\newcommand{\meteg}{g_0}

\newcommand{\metg}{g}

\newcommand{\metgs}{\scale{g}}

\newcommand{\multia}{A}

\newcommand{\multiap}{{A'}}

\newcommand{\multiapp}{{A''}}

\newcommand{\multib}{B}

\newcommand{\multindset}{{\cal J}}

\newcommand{\parvar}{{\bf a}} 

\newcommand{\parvarnbd}{{\cal A}}

\newcommand{\phis}{\scale{\phi}}

\newcommand{\sca}{\alpha}

\newcommand{\scb}{\beta}

\newcommand{\sff}{\chi}
\newcommand{\sffbar}{\underline{\chi}}

\newcommand{\soldom}{X}

\newcommand{\solvreg}{\soldom} 

\newcommand{\sscal}{\eta} 

\newcommand{\sscalup}{\up{\sscal}}

\newcommand{\tframe}{\partial}

\newcommand{\thetas}{{\scale{\theta}}}

\newcommand{\chihats}{{\chihat{}'}}

\newcommand{\eframeb}{\frameb} 

\newcommand{\eframec}{\framec} 

\newcommand{\eframed}{\framed} 
\newcommand{\chihatgo}{{}^g\chihat}
\newcommand{\geodmap}{\Phi}
\newcommand{\Laff}{\lambda}
\newcommand{\curmaxLaff}{\Laff_0}
\newcommand{\id}{\textbf{1}}
\newcommand{\cssscoeff}{C} 
\newcommand{\shearboundsconst}{c} 
\newcommand{\Lafftr}{\Laff^*}
\newcommand{\curustr}{\curus^*}
\newcommand{\curubarstr}{\curubars^*}
\begin{document}
\maketitle
\begin{abstract}
Work of Christodoulou, Klainerman, Rodnianski, Luk, An, and others has provided a number of results on the dynamical formation of trapped surfaces in vacuum solutions to the Einstein field equations. Since the stability of Minkowski spacetime as proved by Christodoulou and Klainerman implies that `small' initial data to the Einstein vacuum equations must give rise to a solution with no singularities, and hence no trapped surfaces, it has been assumed that dynamical formation of trapped surfaces requires a (hard) {\it large-data\/} existence result for the nonlinear hyperbolic Einstein field equations. In this paper we show, to the contrary, that dynamical trapped surface formation results qualitatively similar to those of Christodoulou and Klainerman-Rodnianski can be obtained from (classical) {\it local}, {\it small-data\/} existence results via a scaling we term a {\it radial boost}. In the process we also fully elucidate the {\it short-pulse\/} ansatz as a geometric optics ansatz by showing that, in this scaled picture, the so-called incoming shear satisfies, at highest order, a {\it linear\/} wave equation.
\end{abstract}
\section{Introduction}
In 1916, a year after Eintstein published his field equations, Schwarzschild \cite{schwarzschild} gave his celebrated solution for spherically symmetric spacetimes. This solution featured a curvature singularity. Work of Oppenheimer and Schneider \cite{oppenheimer} showed that, at least in spherical symmetry and in the presence of a collapsing cloud of dust, such a singularity could develop dynamically. Further work of Penrose \cite{penrose} (see also \cite{hawkingEllis}) showed that (in vacuum, or under reasonable conditions on the matter distribution) `singularities', in the form of incomplete causal geodesics, must exist whenever a solution to the Einstein equations possesses a so-called {\it trapped surface}, which is a smooth closed surface such that both ingoing and outgoing congruences of null geodesics converge towards the future. In 2009, Christodoulou \cite{christodoulou} showed that such trapped surfaces could form dynamically in solutions to the Einstein vacuum equations. The bulk of his work was taken up by a proof of a large-data existence result for the Einstein vacuum equations for a specific class of characteristic initial data satisfying a so-called {\it short-pulse\/} ansatz. Subsequently, Klainerman and Rodnianski \cite{klainrodi}, An and Luk \cite{anluk}, and An \cite{an} introduced various scaled and weighted norms in terms of which the proof simplified considerably. However, the question of obtaining an {\it explanation\/} for the success of the short-pulse ansatz, in the sense of finding a setting in which the PDE existence result becomes `expected' in some sense, has remained unexplored.
\par
In this paper we shall take up this question and show, contrary to received wisdom, that in fact `large-data' existence results qualitatively similar to those of Christodoulou and Klainerman-Rodnianski can be obtained by applying a simple scaling, which we term a {\it radial boost}, to solutions obtained from the classical {\it local}, {\it small-data\/} existence result of Rendall \cite{rendall}. We provide a toy model of a spherically symmetric (non-vacuum) spacetime possessing trapped surfaces to illustrate why this construction is reasonable. We shall also discover that the so-called {\it incoming shear tensor}, to which the short-pulse ansatz applies, satisfies, in the scaled picture and up to ignorable small terms, a {\it linear\/} wave equation, and that the short-pulse initial data of Christodoulou and Klainerman-Rodnianski can be viewed as geometric optics initial data for this wave equation.
\par
In the next section we shall give some elementary background in related results in the literature in order to orient the following discussion, and construct a toy model to give some intuition into our main result. In Section 3 we then state and prove our main result, and in Section 4 we give more discussion relating our work to results in the literature. 
\section{Background and motivation} In \cite{christodoulou}, Christodoulou proved essentially the following result. (The full version involves Sobolev space estimates; since one of the main points of the current paper is that one can construct and study these solutions {\it without\/} any knowledge of Sobolev spaces, we do not give these details here.)
\begin{theorem}
Consider characteristic initial data for the Einstein vacuum equations consisting of Minkowski data followed by a `short pulse' of width $\delta$ of incoming gravitational energy. For suitable choices of the incoming energy, such that the incoming gravitational energy is sufficiently large in every direction, the corresponding solution to the Einstein vacuum equations will contain a trapped surface.
\end{theorem}
In their extension of Christodoulou's work, Klainerman and Rodnianski \cite{klainrodi}\ improved this by allowing for a more general class of initial data, including data with large angular derivatives. Subsequent work of An and Luk \cite{anluk} and An \cite{an}\ added decay ans\"atze to the scaling ans\"atze of \cite{christodoulou} and \cite{klainrodi}. These latter results also provided for trapped surface formation from much smaller initial data, but at the requirement of much greater angular smoothness.
\par
All of the foregoing results require sophisticated energy estimates, based on the so-called {\it null structure equations}, to prove large-data existence results.\footnote{By `large-data' we mean here data which is not close to Minkowski data in the chosen gauge. While the results of \cite{christodoulou} and \cite{klainrodi} are most clearly large-data in this sense, the results in \cite{anluk} {\it become\/} so (though they are not initially) because of the region over which they solve. We thank Peng Zhao for bringing this point to our attention.} More precisely, it can be shown that the connection coefficients and various null curvature components of a Lorentzian metric satisfying the Einstein vacuum equations must satisfy a coupled system of transport and constraint equations\footnote{These were apparently first derived in unpublished notes of Christodoulou.}; see Klainerman and Nicol\'o \cite{klainnicolo}\ for a sketch of the derivation and Christodoulou \cite{christodoulou}\ for a detailed treatment. These equations, however, are not solved directly but rather used in concert with local existence results of Rendall \cite{rendall} and Choquet-Bruhat (beginning with, e.g., \cite{choquetbruhat}) to prove existence via a bootstrap argument. More precisely, Christodoulou \cite{christodoulou} shows existence of the desired solution in an initial domain by an appeal to Rendall \cite{rendall}, and then shows local continuation, given the above-mentioned energy estimates, via Choquet-Bruhat \cite{choquetbruhat}. The above-cited papers work within the overall context established by \cite{christodoulou} (in particular, implicitly relying on \cite{rendall} and \cite{choquetbruhat}) and focus only on refining the necessary energy estimates. In other words, all of the works cited above rely implicitly on (a) the derivation of the null structure equations and (b) the local existence and continuation arguments from \cite{christodoulou}.
\par
In Christodoulou's original result, the initial incoming shear was explicitly required to satisfy an ansatz of the form
\begin{equation}
\chihatz = \delta^{-1/2} f(\delta^{-1} \chrubar, \chrangulvar),\label{chrchihatzansatz}
\end{equation}
where $\chrubar$ is an affine parameter along the null geodesics ruling the outgoing initial null hypersurface, and $\chrangulvar \in \Sph^2$. Klainerman and Rodnianski \cite{klainrodi} observe that the following extension of this ansatz is natural\footnote{\cite{klainrodi} relate this to the parabolic scaling of the Minkowski wave equation, though given the angular character of $\krangulvar$ this analogy seems somewhat strained in their setting.}:
\begin{equation}
\chihatz = \delta^{-1/2} f(\delta^{-1} \krubar, \delta^{-1/2} \krangulvar),\label{krchihatzansatz}
\end{equation}
but since $\krangulvar \in \Sph^2$ it is unclear how to apply this ansatz directly. They therefore settle for norm ans\"atze inspired by, but weaker than, those that would obtain given initial data of the form (\ref{krchihatzansatz}).
\par
A scaling similar to that in (\ref{krchihatzansatz}) was independently rediscovered much later, in a different context, by the present author in his doctoral thesis \cite{ntcthesis} (see also \cite{spntcarXiv}). In \cite{ntcthesis} highly concentrated solutions to the Einstein vacuum equations were sought under the additional assumption of $U(1)$ (translational) symmetry. As is well-known, in this setting, the Einstein vacuum equations reduce to a system of Riccati ODE equations coupled to a linear wave equation. In \cite{ntcthesis} the geometry of the spatial cross-section was taken to be {\it rectangular} (in that case, the entire real line $\R^1$) and it was observed that an ansatz which would correspond roughly to\footnote{\cite{ntcthesis} worked with the metric components directly, and the initial data was given in terms of the metric, not the shear. Further confounding comparisons, \cite{ntcthesis} worked with a {\it large\/} parameter $k$ instead of a {\it small\/} parameter $\delta$.}
\begin{equation}
\chihatz = \delta^{-1/2} f(\delta^{-1} \ntcubar, \delta^{-1/2} \ntcx)\label{ntcchihatzansatz}
\end{equation}
transformed this system to one in which all nonlinear terms in the wave equation were {\it small}. More precisely, in \cite{ntcthesis} a {\it null geodesic gauge} $\ntcu$, $\ntcubar$, $\ntcx \in \R^1$ was used, as opposed to the double-null gauge of previous workers, and a {\it coordinate transformation}
\begin{equation}
\ntcu \mapsto \ntcu,\,\ntcubar \mapsto \delta^{-1} \ntcubar,\, \ntcx \mapsto \delta^{-1/2}\ntcx\label{ntcscal}
\end{equation}
was applied; in the new coordinates, all nonlinear terms in the wave equation can be shown to be small.
\par
We note that (\ref{ntcscal}) is a Lorentz boost with parameter $\delta^{-1/2}$, i.e., $\ntcu \mapsto \delta^{1/2} \ntcu$, $\ntcubar \mapsto \delta^{-1/2} \ntcubar$, followed by an isotropic scaling by $\delta^{-1/2}$.
\par
It is not hard to see that a scaling like that in (\ref{ntcscal}), but without the transverse spatial scaling, would reduce the {\it pointwise\/} size of the incoming shear in (\ref{chrchihatzansatz}) from $\sim \delta^{-1/2}$ to $\sim 1$.\footnote{This was, in particular, clear to the present author as far back as 2021.} On the other hand, the detailed use of the coordinate scaling in \cite{ntcthesis} relies heavily on an additional conformal rescaling (by $\delta^{-1}$) of the Minkowski metric and the concomitant transverse spatial scaling, and makes implicit use of the form of the Minkowski metric in rectangular coordinates. In greater detail, under (\ref{ntcscal}) and an overall conformal scaling by $\delta^{-1}$, the Minkowski metric $\eta = -d\ntcu \otimes d\ntcubar - d\ntcubar \otimes d\ntcu + d\ntcx\otimes d\ntcx$ transforms to
\begin{equation}
\eta' = \delta^{-1} \left[ -d\ntcu \otimes d(\delta \ntcubar) - d(\delta \ntcubar) \otimes d\ntcu + d(\delta^{1/2} \ntcx) \otimes d(\delta^{1/2} \ntcx)\right] = \eta,
\end{equation}
while if we were to implement on the Minkowski metric in (say) the form
\begin{equation}
\eta = -d\curu \otimes d\curubar - d\curubar \otimes d\curu + [ \currz + 2^{-1/2} (\curubar - \curu)]^2 (d\curtheta \otimes d\curtheta + \sin^2(\curtheta) d\curphi \otimes d\curphi)\label{initetasphform}
\end{equation}
the scaling (cf.\ (\ref{krchihatzansatz}))
\begin{equation}
\curu \mapsto \curu,\,\curubar \mapsto \delta^{-1}\curubar, \curangulvar \mapsto \delta^{-1/2} \curangulvar,\label{curscal}
\end{equation}
$\curangulvar = (\curtheta, \curphi)$, together with an overall conformal scaling by $\delta^{-1}$, we would obtain instead the metric
\begin{equation}
\etap 
= -d\curu \otimes d\curubar - d\curubar \otimes d\curu\\
{}+ [\currz + 2^{-1/2} (\delta\curubar - \curu)]^2 (d\curtheta \otimes d\curtheta + \sin^2 (\delta^{1/2} \curtheta) d\curphi \otimes d\curphi),\label{scalinitetasphform}
\end{equation}
which is not at all the same form as (\ref{initetasphform}) and, worse, becomes degenerate as $\delta \rightarrow 0$. It is thus unclear how to extend the coordinate scaling to the spherical setting.
\par
Further complicating matters is the result of Luk and Moschidis \cite{lukmoschidis} which shows that hypersurfaces in a null vacuum spacetime which are `sufficiently close' to any spacelike hypersurface in Minkowski spacetime cannot contain trapped surfaces. Since the scaling in (\ref{ntcscal}), together with an overall conformal scaling (which does not change any causality properties), preserves the Minkowski metric, this would suggest that no spacelike hypersurface in solutions obtained via (\ref{ntcscal}) can contain a trapped surface. One might expect that a similar problem would arise when attempting to study trapped surface formation via the scaling (\ref{curscal}).
\par
On the other hand, it is clear that if we replace $\theta$ in (\ref{initetasphform}) -- (\ref{scalinitetasphform}) by $\theta - \pi/2$ (i.e., measure our angles from the equator rather than the poles), then the scaled metric in (\ref{scalinitetasphform}) is instead
\begin{equation}
\etap = -d\curu \otimes d\curubar - d\curubar \otimes d\curu + [\currz + 2^{-1/2} (\delta\curubar - \curu)]^2 (d\curtheta \otimes d\curtheta + \cos^2 (\delta^{1/2} \curtheta) d\curphi \otimes d\curphi),\label{modscalinitetasphform}
\end{equation}
which is easily seen to be an $O(\delta)$ perturbation of the (flat) metric obtained by taking the limit $\delta \rightarrow 0$, and which we term the {\it infinitely radially boosted Minkowski metric}:
\begin{equation}\label{infboostmink}
\etazp = -d\curu \otimes d\curubar - d\curubar \otimes d\curu + [\currz - 2^{-1/2} \curu]^2 (d\curtheta \otimes d\curtheta + d\curphi \otimes d\curphi).
\end{equation}
We note that the outgoing null expansion of any symmetry sphere of $\etazp$ {\it vanishes}, suggesting that $\etazp$ is {\it unstable\/} (in the full space of all Lorentzian metrics, which may or may not satisfy the Einstein vacuum equations) against the formation of trapped surfaces. While we shall not study the geometry of $\etazp$ in detail, we do note that it possesses a codimension-1 singular hypersurface at $\curu = 2^{1/2} \currz$, and that all timelike future-directed geodesics must intersect this hypersurface.
\par
Here and below we are being intentionally vague about the meaning of terms such as `small perturbation'. In particular, the way we use this term is generally {\it coordinate-dependent}. See Section \ref{discussion} for further discussion of this point.
\par
\par
Motivated by this last observation, we construct toy models of metrics which possess closed trapped surfaces despite being, in the neighborhood of at least one trapped surface, an $O(\delta)$ perturbation of (\ref{modscalinitetasphform}). Since these models are used only for motivation we will not try to make our considerations entirely precise.
We first recall the causal structure of Schwarzschild spacetime; see, e.g., \cite{hawkingEllis}, Figure 23, and surrounding discussion. One way of explaining the fact that any future-directed causal geodesic crossing the event horizon must ultimately reach the singularity at $r = 0$ is to observe that, when crossing the event horizon, the future-directed null cones `tip over' sufficiently that all future-directed causal vectors point towards decreasing $r$. Of course, in the Schwarzschild solution the required amount of `tipping' is sufficient that the Schwarzschild metric in the above coordinates is by no means a `small' perturbation of Minkowski.
On the other hand, defining new `radially boosted' coordinates (see (\ref{curscal}))
\begin{equation}
\curup = \curu,\,\curubarp = \delta^{-1}\curubar, \curangulvarp = \delta^{-1/2} \curangulvar,\label{curscaleq}
\end{equation}
and concomitant time and radial coordinates
\begin{equation}
\curt = 2^{-1/2} (\curubar + \curu),\,\curdr = 2^{-1/2} (\curubar - \curu),\,\curtp = 2^{-1/2} (\curubarp + \curup),\,\curdrp = 2^{-1/2} (\curubarp - \curup),
\end{equation}
a spacetime diagram of a $\curu\curubar$ cross-section, in $\curup\curubarp$ coordinates, will be roughly as in Figure \ref{radialboostfig}. There the blue lines repesent null cones in Minkowski spacetime and the red lines represent a slightly `tilted' version. It is thus clear graphically that, in the $\curup\curubarp$ coordinates, the null cones need only `tilt' an amount $O(\delta)$ for all future-directed causal vectors to point towards decreasing $r$.
\par
\begin{figure}
\centering
\includegraphics[keepaspectratio]{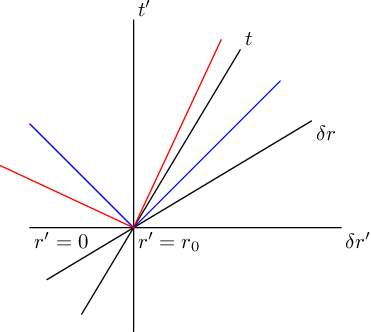}
\caption{Radially boosted Minkowski spacetime}\label{radialboostfig}
\end{figure}
\par
This may all be made precise as follows. We define a new metric $\etarot$, the angular part of which equals the Minkowski metric, but has a transverse part 
\begin{equation}
- d\cururot \otimes d\curubarrot - d\curubarrot \otimes d\cururot,
\end{equation}
where
\begin{equation}
\cururot = \cos\rotangle \curup - \sin\rotangle \curubarp,\,\curubarrot = \sin\rotangle \curup + \cos\rotangle \curubarp.
\end{equation}
The metric $\etarot$ is thus given by
\begin{align}
\etarot &= - d\cururot \otimes d\curubarrot - d\curubarrot \otimes d\cururot\notag\\
{}&\qquad + \left[ \currz + 2^{-1/2} (\delta \curubarp - \curup) \right]^2 (d\curthetap \otimes d\curthetap + \cos^2 (\delta^{1/2} \curthetap) d\curphip \otimes d\curphip)\\
&= -d\cururot \otimes d\curubarrot - d\curubarrot \otimes d\cururot\notag\\
{}&\qquad + \left[ \currz + 2^{-1/2} \left\{-\delta \sin\rotangle - \cos\rotangle) \cururot + (\delta \cos\rotangle - \sin\rotangle) \curubarrot\right\}\right]^2\notag\\
{}&\qquad\quad \cdot (d\curthetap \otimes d\curthetap + \cos^2 (\delta^{1/2} \curthetap) d\curphip \otimes d\curphip).
\end{align}
Defining $\Delta = \vert\det\,\etarot\vert$, the incoming and outgoing null expansions of this metric at a point $(\cururot, \curubarrot)$ are thus (up to a numerical factor) 
\begin{equation}
\partial_\cururot \log\Delta = \Delta^{-1} (-\delta \sin\rotangle - \cos\rotangle),\,\partial_\curubarrot \log\Delta = \Delta^{-1} (\delta \cos\rotangle - \sin\rotangle).
\end{equation}
Thus the point $(\cururot, \curubarrot)$ will represent a trapped sphere for $\etarot$ when $\rotangle$ satisfies
\begin{equation}
\tan \rotangle > \delta.
\end{equation}
Since $\tan\rotangle \sim \delta$ gives, {\it in the scaled coordinate system}, $\etarot \sim \eta + O(\delta)$, this shows that, in an appropriate coordinate system, a `small' perturbation of the Minkowski metric can lead to a trapped surface.
\par
Despite its limitations, the foregoing toy model gives the following insights which will be sufficient to allow us to overcome its limitations:
\begin{enumerate}
\item The finitely radially boosted Minkowski metric $\etap$ is an $O(\delta)$ perturbation of the infinitely radially boosted Minkowski metric $\etaib$.
\item Trapped surfaces can appear after a further $O(\delta)$ perturbation of $\etap$.
\end{enumerate}
Our results below will give trapped surface formation through the mechanism discovered by Christodoulou, which uses the Raychaudhuri equation (\cite{hawkingEllis})
\begin{equation}
\partial_\curubar \tr\,\schi = -\frac{1}{2} (\tr\schi)^2 - \left\vert\chihat\right\vert^2
\end{equation}
together with specially chosen data for $\chihat$ which naturally concentrates along spheres on ingoing null hypersurfaces. By the foregoing, we expect to require $\left\vert\chihat\right\vert^2 \sim \delta$, and thus to require a shear $\chihat \sim \delta^{1/2}$, which after adjusting for the radial boost is precisely as in Christodoulou's results. See (\ref{chihaticondb}), (\ref{chihatsicondb}).
\par
We close this section by observing that one way of reconciling the solutions to be constructed below with the nonexistence result of Luk and Moschidis \cite{lukmoschidis} mentioned above is to note that the `closeness' of the given hypersurface to a spacelike hypersurface in Minkowski space required by \cite{lukmoschidis} depends on how `close to null' the given hypersurface is. In our setting, because of the large radial boost, the natural spacelike sections $\curtp = \textrm{constant}$, in which our trapped surfaces are located, are very close to null (and in fact become null in the limit $\delta \rightarrow 0$).
\par
\section{Main result} We shall now show that the informal considerations in the previous section can be carried through to produce solutions to the Einstein vacuum equations exhibiting dynamic formation of trapped surfaces. Unlike most previous results in the literature, we do not rely on either the derivation of the null structure equations or the local existence/continuation results in \cite{christodoulou} but proceed from first principles, relying only on the local existence result in \cite{rendall} which is also at the foundation of the work in \cite{christodoulou}.
\par
In the study of partial differential equations, notions of `size' and `closeness' are typically formalized in terms of various Sobolev norms (and received wisdom holds that, in general, partial differential equations can {\it only\/} be solved in terms of $L^2$ Sobolev norms). In this paper we do not need to {\it solve\/} any equations but only {\it estimate\/} various quantities obtained from already-known solutions. Further, the only `closeness' or approximation results available from \cite{rendall} are in terms of smoothly parameterized families of functions. Thus the notion of `closeness' we shall use is that embodied in the following definition.
\par
\begin{definition}\label{Osizedef} Let $f : \freg_1 \times \freg_2 \subset \R^p \times \R^q \rightarrow \R^k$, $\freg_1$, $\freg_2$ open, $0 \in \freg_1$, be smooth. Let $\parvar$, $\datvar$ denote arbitrary elements of $\freg_1$, $\freg_2$. Let $\multia$ be a $p$-multiindex. Suppose that there is a smooth function $g : \freg_1 \times \freg_2 \rightarrow \R^k$ such that $f(\parvar, \datvar) = \parvar^\multia g(\parvar, \datvar)$. Then we write $f = \Osize(\parvar^\multia)$. If $f_1$, $f_2$ are such that $f_1 - f_2 = \Osize(\parvar^\multia)$, then we write $f_1 = f_2 + \Osize(\parvar^\multia)$. In cases where we wish to emphasize the coordinates on $\freg_2$ we shall write, e.g., $f = \Osize(\parvar^\multia)$.
\end{definition}
\par\noindent
As usual, to be precise we should really define $\Osize(\parvar^\multia)$ to be the set of all functions satisfying the above condition (or, more generally, with germ at $\parvar = 0$ satisfying the above condition), and write $f \in \Osize(\parvar^\multia)$, but we shall stick with the above notation for convenience. We note though that if $\multia < \multiap$, then $\Osize(\parvar^\multiap) = \Osize(\parvar^\multia)$, but $\Osize(\parvar^\multia) \neq \Osize(\parvar^\multiap)$. 
\par
We have the following results.
\par
\begin{proposition}\label{Osizechar} Let $\multia_i \neq 0$ be disjoint multiindices. Then $f = \sum_{i = 1}^k \Osize(\parvar^{\multia_i})$ if and only if the following holds: let $\{ \multib_i \}$ be any collection of multiindices satisfying $\multib_i < \multia_i$, and let $\multindset = \{ j : \multia_{ij} - \multib_{ij} > 0 \text{ for some } i \}$; then
\begin{equation}\label{Osizecharcond}
\left(\prod_{i = 1}^k \partial_\parvar^{\multib_i}\right) f|_{\{ \parvar : \parvar_j = 0 \text{ for all } j \in \multindset \}} = 0.
\end{equation}
\end{proposition}
\begin{proof}
The forward direction is clear since the $\multia_i$ are disjoint. Now suppose that the condition above holds, let $\{ \multib_i \}$ be as stated, and set $\multindset_i = \{ j : \multia_{ij} - \multib_{ij} > 0 \}$. Assume without loss of generality that $\multindset_1 \neq \emptyset$. Let $f_1$ be the function obtained from $f$ by setting $\alpha_j = 0$ for all $j \in \multindset_1$. We claim that $f - f_1 = \Osize(\parvar^{\multia_1})$. This may be seen as follows. Let $j_0 = \min \multindset_1$. By (\ref{Osizecharcond}), $f - f_1 = \parvar_{j_0}^{A_{1j_0}} g_1$ for some smooth function $g_1$. By induction, we thus obtain $f - f_1 = \Osize(\parvar^{\multia_1})$. Now (\ref{Osizecharcond}) with $i$ running from 2 to $k$ clearly applies to $f_1$, and the result follows by a further induction.
\end{proof}
%
\par\noindent
Non-disjoint indices can be handled by noting that, e.g., if $\multia = \inf \{ \multia_1, \multia_2 \}$, then $f = \Osize(\parvar^{\multia_1}) + \Osize(\parvar^{\multia_2})$ if and only if $f = \Osize(\parvar^\multia) [\Osize(\parvar^{\multia_1 - \multia}) + \Osize(\parvar^{\multia_2 - \multia})]$, where $\multia_1 - \multia$ and $\multia_2 - \multia$ are disjoint.
\par
We also have the following two results.
\par
\begin{lemma}\label{Osizecompose} If $f = f_0 + \Osize(\parvar^\multia)$, and $F : \freg \subset \R^k \rightarrow \R^\ell$ is smooth on a region containing the range of $f$ and $f_0$, then $F(f(\parvar, \datvar)) - F(f_0(\parvar, \datvar)) = \Osize(\parvar^\multia)$. Moreover, for any derivative operator ${\cal D}$ in $\datvar$, ${\cal D} f = {\cal D} f_0 + \Osize(\parvar^\multia)$.
\end{lemma}
\begin{proof} The first part follows from Hadamard's lemma (see, e.g., \cite{bierstone}, Lemma 1.2): $F(f(\parvar, \datvar)) - F(f_0(\parvar, \datvar)) = (f - f_0) \cdot g(f, f_0)$. The second part follows directly from Definition \ref{Osizedef}.\end{proof}
\par
While the first part of Lemma \ref{Osizecompose} has an analogue in $L^2$ Sobolev space theory, the second most definitely does not. We note that it is this second part of Lemma \ref{Osizecompose} which makes the $\Osize$ notation truly useful for performing estimates. 
\par
\begin{lemma}\label{OsizeODE} Suppose that $F : \freg_1 \times \freg_2 \times \interval \subset \R^p \times \R^q \times \R^1 \rightarrow \R^q$ satisfies $F(\parvar, \datvar, s) = F^0(\datvar, s) + \Osize(\parvar^{\multiap})$, and let $\ddatvar : \freg_1 \rightarrow \R^q$. Then there are regions $\freg'_1 \subset \freg_1$, $\freg'_2 \subset \freg_2$, and an interval $\intervalp \subset \interval$ such that the following equations have unique solutions on $\intervalp$ for all $\parvar \in \freg'_1$, $\datvarii \in \freg'_2$:
\begin{align}
\dot{\datvari} = F^0(\datvari, s),&\quad \datvari(0) = \datvarii\\
\dot{\datvar} = F(\parvar, \datvar, s),&\quad \datvar(0) = \datvarii + \ddatvar(\parvar).\label{OsizeODEfsys}
\end{align}
Moreover, for $\datvar_{\parvar, \datvarii}(s)$ the solution to (\ref{OsizeODEfsys}), the map $\freg'_1 \times \freg'_2 \times \intervalp \rightarrow \R^q$, $(\parvar, \datvarii, s) \mapsto \datvar_{\parvar, \datvarii}(s)$ is smooth. Finally, if $\ddatvar = \Osize(\parvar^\multiapp)$ and $\multia = \inf \{ \multiap, \multiapp \}$, then $\datvar_{\parvar, \datvarii} = \datvari_\datvarii + \Osize(\parvar^\multia)$.
\end{lemma}
\begin{proof} The existence and smoothness results follow from standard arguments (see, e.g., Sideris \cite{sideris}, Chapter 6). The estimate follows by noting that if $\multib < \multia$, then
\begin{equation}
\frac{\partial}{\partial s} \left(\partial_\parvar^\multib \datvar\right) = \partial_\parvar^\multib \Osize(\parvar^\multiap),\, (\partial_\parvar^\multib \datvar)(\parvar, 0) = \partial_\parvar^\multib \Osize(\parvar^\multiapp),
\end{equation}
both of which vanish at $\parvar_i = 0$ whenever $\multib_i < \multia_i$, and applying Proposition \ref{Osizechar}. \end{proof}
\par
We shall work with the same general geometric picture as Christodoulou \cite{christodoulou} and subsequent workers; thus we shall assume that our initial data is given on incoming and outgoing null hypersurfaces intersecting in a topological sphere. We shall moreover assume that the spacetime is Minkowskian everywhere to the past of the incoming null hypersurface, and assign nontrivial data only on the outgoing null hypersurface.
\par
The existence result of Rendall \cite{rendall}, which we discuss next, is based on a very similar geometric setup except that the intersection of the two hypersurfaces is initially assumed to be covered by a single coordinate patch. Moreover, the metric is constructed in so-called {\it harmonic\/} (or {\it wave\/}) coordinates, which satisfy
\begin{equation}
\metg^{ij} \nabla_i \nabla_j x^k = 0.
\end{equation}
Specifically, Rendall proves the following.
\par
\begin{theorem}\label{rendallexistence} Let $x^i$, $i = 0, 1, 2, 3$, be coordinates on $\R^4$, and let $\init_0 = \{ x^1 = 0 \}$, $\init_1 = \{ x^0 = 0 \}$. Suppose given on $\init_0 \cup \init_1$ a smooth, symmetric positive-definite matrix $\initmet_{\sca\scb}$, $\sca, \scb \in \{ 2, 3 \}$, and suppose given on $\initout \cap \initin$ smooth functions $\conffact$, $\conffact_0$, $\conffact_1$, $\metd_0$, $\metd_1$. Then there is a neighborhood $U$ of $\initout \cap \initin$ in the quadrant $\{ x^0 \geq 0, x^1 \geq 0 \}$ on which there is a unique smooth solution to the Einstein vacuum equations $\metg$ and a smooth function $\conffact$, agreeing with that given on $\init_0 \cap \init_1$, such that on $\init_0 \cup \init_1$, $\metg_{\sca\scb} = \Omega \initmet_{\sca\scb}$, while on $\init_0 \cap \init_1$ we have $\partial_0 \conffact = \conffact_0$, $\partial_1 \conffact = \conffact_1$, $\partial_0 \metg_{13} = \metd_0$, $\partial_0 \metg_{14} = \metd_1$.
\par
Furthermore, if the initial data $\initmet$, $\conffact$, $\conffact_i$, $\metd_i$ depend smoothly on a finite number of parameters, say $\parvar$, in a neighborhood of some $\parvar_0$, then the following holds: there is a neighborhood $\parvarnbd$ of $\parvar_0$ and a neighborhood $U'$ of $\init_0 \cap \init_1$ such that the foregoing holds on $U'$ for all $\parvar \in \parvarnbd$, and the solution $\metg$ and conformal factor $\conffact$ depend smoothly on $\parvar$. \end{theorem}
\par\noindent
When initial data is given on more than one coordinate chart, it is shown that the null harmonic coordinates developed from the overlapping portions of the charts must agree on the intersection of the domain of dependence of the charts. See the discussion in \cite{rendall}, 2.2.
\par
We shall write $\Sph^2$ for the unit sphere in $\R^3$ with the standard induced metric. For any $\curangulvar \in \Sph^2$ we define an {\it adapted spherical coordinate system\/} $(\theta, \phi)$ on $\Sph^2$ around $\curangulvar$ to be the spherical coordinate system on $\Sph^2$ positioned so that $\curangulvar$ coincides with the point $(1, 0, 0)$, and satisfying $\theta(\curangulvar) = 0$, $\phi(\curangulvar) = 0$; thus the metric on $\Sph^2$ in the coordinates $(\theta, \phi)$ is given by
$$
d\theta^2 + \cos^2 \theta d\phi^2.
$$
\par
Given a (time-orientable) Lorentzian metric $\meteg$ on some open set $U$ in $\R^4$ and a spacelike 2-surface $\Sigma \subset U$, we define the {\it null second fundamental forms\/} of $\meteg$ at a point $x \in \Sigma$ as follows. Choose future-directed null vectors $\L$ and $\Lbar$ at $x$ which are normal to $T_x \Sigma$ and satisfy $g(\L, \Lbar) = -1$;\footnote{There are multiple different normalization conventions in the literature; \cite{christodoulou} uses three distinct ones for various purposes. The one chosen here is convenient for our purposes since it allows us to take $\L$ and $\Lbar$ to be small perturbations of coordinate vector fields.} then the null second fundamental forms of $\meteg$ with respect to $\L$, $\Lbar$ are (see, e.g., \cite{oneill}, Chapter 4), for $\tframe_\sca$, $\tframe_\scb$ a frame on $\Sigma$,
\begin{equation}
\sff_{\sca\scb} = -\meteg(\L, \nabla_{\tframe_\sca} \tframe_\scb),\quad \sffbar_{\sca\scb} = -\meteg(\Lbar, \nabla_{\tframe_\sca} \tframe_\scb);
\end{equation}
if there are coordinates $\curu$, $\curubar$ such that $\L = \partial_\curubar$, $\Lbar = \partial_\curu$, then we also have
\begin{equation}\label{sffmetexp}
\sff_{\sca\scb} = \frac{1}{2} \partial_\curubar \meteg{}_{\sca\scb},\quad \sffbar_{\sca\scb} = \frac{1}{2} \partial_\curu \meteg{}_{\sca\scb}.
\end{equation}
In an orthonormal frame $\{ \eframe_\framea \}$, we also have
\begin{equation}
\sff_{\framea\frameb} = \meteg(\nabla_{\eframe_\framea} \L, \eframe_\frameb), \sffbar_{\framea\frameb} = \meteg(\nabla_{\eframe_\framea} \Lbar, \eframe_\frameb).
\end{equation}
Since $\L$ and $\Lbar$ are unique only up to the ({\it infinitesimal\/} boost) rescaling $\L \mapsto \alpha \L$, $\Lbar \mapsto \alpha^{-1} \Lbar$, $\sff$ and $\sffbar$ are likewise defined only up to the same rescaling. (See, e.g., \cite{klainnicolo}, (3.1.8).) Thus $\sff$ and $\sffbar$ are properly functions of a metric {\it and\/} a choice of null pair $(\L, \Lbar)$. In our work below, we shall use the size of the underlying {\it metric\/} to measure the size of the initial data.
\par
Next, we note some properties of the infinitely radially boosted Minkowski metric described in the last section (see (\ref{infboostmink})), which we write as follows:
\begin{equation}
\etaib = -d\curup \otimes d\curubarp - d\curubarp \otimes d\curup + [\currz - 2^{-1/2} \curup]^2 (d\curthetap \otimes d\curthetap + d\curphip \otimes d\curphip).
\end{equation}
\par
\begin{lemma}\label{etaibprops} Define
\begin{equation}
\Lib = \partial_\curubarp,\,\Libbar = \partial_\curup,\,\eframe_1 = 1/(\currz - 2^{-1/2} \curup) \partial_\curthetap,\,\eframe_2 = 1/(\currz - 2^{-1/2} \curup) \partial_\curphip.\label{etaibframedef}
\end{equation}
Then (for $\frameinda, \frameindb \in \{ 1, 2 \}$, $\sca, \scb \in \{ 2, 3 \}$) 
\begin{gather}
\nabla_\Libbar \eframe_\frameinda = 0,\quad \nabla \Lib = 0,\quad \nabla_\Libbar \Libbar = 0,\quad [\Lib, \Libbar ] = 0,\\
\nabla_{\eframe_\frameinda} \Libbar = -\frac{1}{\sqrt{2}(\currz - 2^{-1/2}\curup)} \eframe_\frameinda,\quad [\Libbar, \eframe_\frameinda] = \frac{1}{\sqrt{2}(\currz - \curup)} \eframe_\frameinda,\\
\nabla_{\eframe_\frameinda} \eframe_\frameindb = -\frac{1}{\sqrt{2} (\currz - 2^{-1/2}\curup)} \delta_{AB} \Lib,\quad [\eframe_\frameinda, \eframe_\frameindb] = 0,\\
\chi_{\sca\scb} = 0,\quad \chibar_{\sca\scb} = -\frac{1}{\sqrt{2}(\currz - 2^{-1/2} \curup)}.
\end{gather}
\end{lemma}
\begin{proof} These all follow by straightforward computations. \end{proof}
\par
\par
Finally, consider the region $[0, \curu_0) \times [0, \curubar_0) \times \Sph^2$ in $\R^4$, let $\curangulvar \in \Sph^2$, and let $(\theta, \phi)$ be an adapted coordinate system on $\Sph^2$ at $\curangulvar$. We wish to implement the scaling transformation described in Section 2 above (see \ref{curscaleq}). Thus, for any $\cscal > 0$, we define {\it scaled coordinates\/} $(\curus, \curubars, \thetas, \phis)$ as follows:
\begin{equation}
\curus = \curu,\,\curubars = \cscal^{-1}\curubar, \curangulvars = \cscal^{-1/2} \curangulvar.\label{curscaleqredux}
\end{equation}
In these coordinates, the Minkowski metric has the form (see (\ref{modscalinitetasphform}))
\begin{equation}
\etap = -d\curus \otimes d\curubars - d\curubars \otimes d\curus + [\currz + 2^{-1/2} (\cscal\curubars - \curus)]^2 (d\thetas \otimes d\thetas + \cos^2 (\cscal^{1/2} \thetas) d\phis \otimes d\phis).\label{modscaletap}
\end{equation}
\par
With these preliminaries out of the way, we are now ready to state and prove our main result.
\par
\begin{theorem} Let $\chihati_{ij}(\cscal^{1/2}, \sscal, \dscal; \curubar, \curangulvar)$, $\cscal \in (-\cscalup, \cscalup)$, $\sscal \in (-\sscalup, \sscalup)$, $\dscal \in (-\dscalup, \dscalup)$, $\curubar \in (-\cscal, \cscal)$ be a smoothly parameterized family of smooth, trace-free, rank-2 covariant tensors on $(-\cscal, \cscal) \times \Sph^2$ which are tangent to the submanifolds $\{ \curubar \} \times S^2$. Suppose that $\chihati$ satisfies the following conditions:
\begin{enumerate}
\item $\chihati = 0$ for $\curubar \leq 0$.
\item Let $\curangulvar \in \Sph^2$, and let $(\theta, \phi)$ be an adapted coordinate system at $\curangulvar$. Then on a neighborhood of $0$ in $(\theta, \phi)$-space of size $\cscal^{1/2}$,
\begin{equation}
\cscal \chihati_{jk} (\cscal^{1/2}, \sscal, \dscal; \curubar, \curangulvar) = \Osize(\cscal)(\curubars, \curangulvars) + \Osize(\sscal)(\curubars, \curangulvars).\label{chihaticondb}
\end{equation}
\end{enumerate}
Then there is a unique solution to the Einstein vacuum equations $\metg$ on a neighborhood $(-\curuf, \curuf) \times [0, \cscal\cdot\curubarf) \times \Sph^2$ of $\init_0 \cap \init_1$, for some $\curuf, \curubarf > 0$ independent of $\cscal$, $\sscal$, and $\dscal$, which is Minkowskian on $\curubar \leq 0$. 
If, moreover, $\chihati$ satisfies, for some $\curustr \in (0, \curuf)$, $\curubarstr \in (0, \cscal \cdot \curubarf)$, for every $\curangulvar \in \Sph^2$ and on a neighborhood of $0$ of size $\cscal^{1/2}$ in an adapted coordinate system $(\theta, \phi)$, and with $\cscal = \cssscoeff \sscal^2$ for some constant $\cssscoeff$,
\begin{enumerate}
\item \begin{equation}\label{angdiffbounds} \cscal^2 (\partial_\theta^2 + \partial_\phi^2) \chihati|_{\theta = \phi = 0} = \Osize(\dscal\sscal), \end{equation}
\item \begin{equation}\label{shearbounds} \frac{1}{\currz} - \frac{\curuf}{\currz^2} + \shearboundsconst \leq \frac{1}{2} \int_0^{\cscal\cdot \curubarf} |\chihati|^2 (0, s, \theta, \phi)\,ds \leq \frac{1}{\currz} - \shearboundsconst, \end{equation}
\end{enumerate}
for some constant $\shearboundsconst$, then $\init_1$ will be free of trapped surfaces but the development $\metg$ will contain a trapped surface near the sphere $\curu =\curustr$, $\curubar = \cscal\cdot\curubarstr$.
\end{theorem}
\par
We write $\chihati$ as a function of $\cscal^{1/2}$ because (see (\ref{modscaletap})) the Minkowskian portion of the initial data is only smooth in $\cscal^{1/2}$, not in $\cscal$, while the literature works in terms of $\cscal$ and we write $\cscal$ far more often than $\cscal^{1/2}$.\footnote{This is a minor technical point with no practical importance.} We note that the aim of the somewhat odd condition (\ref{chihaticondb}) is to allow us to derive the bound in (\ref{chihatsicondb}) (below) on the shear in the {\it scaled\/} coordinates, where it is more natural. The smallness condition in (\ref{angdiffbounds}) should be compared to the auxiliary smallness condition in \cite{klainrodi}, (31). Further, the use of symmetric intervals in the above statement is for simplicity only; all that is required is some interval containing 0. We point out that the initial data constructed in \cite{christodoulou}, Chapter 2, satisfies (\ref{chihaticondb}) with $\cscal = \sscal^2$, and satisfies (\ref{angdiffbounds}) with $\dscal = \cscal$; in fact, in scaled coordinates, it is not hard to see that his metric $m$ is $\id + \Osize(\delta^{1/2})$, while the conformal factor $\phi$ is $1 + \Osize(\alpha)$. This data, of course, has much smaller angular derivatives than are actually required.
\par
\begin{proof} We proceed as follows. For each $\curangulvar \in \Sph^2$ we will show that there is a neighborhood $U \subset \Sph^2$ of $\curangulvar$ and values $\curusff(\curangulvar)$, $\curubarsff(\curangulvar)$ such that the Einstein vacuum equations have a solution $\metg$ on $\soldom = [0, \curusff(\curangulvar)) \times [0, \curubarsff(\curangulvar)) \times U$ realizing the given initial data, in coordinates $(\curus, \curubars, \thetas, \phis)$ which are harmonic for $\metg$. Moreover, as discussed in Rendall \cite{rendall}, Section 5, for $\curangulvar' \in \Sph^2$ with corresponding neighborhood $U'$, region $\soldom'$ and coordinates $(\curus', \curubars', \thetas', \phis')$ such that the corresponding neighborhood $U'$ meets $U$, the coordinates $(\curus, \curubars)$ and $(\curus', \curubars')$ will agree on $\soldom \cap \soldom'$. Since $\Sph^2$ is compact, this will establish the first part of the theorem.
\par
Thus let $\curangulvar \in \Sph^2$, let $(\theta, \phi)$ be an adapted coordinate system around $\curangulvar$, and define scaled coordinates $(\curus, \curubars, \curangulvars)$ by (\ref{curscaleqredux}). We define the scaled shear tensor by (cf.\ (\ref{sffmetexp})) 
\begin{equation}
\chihatsi_{\sca\scb} = \frac{1}{2} \partial_\curubars \metgs_{\sca\scb} = \frac{1}{2} \cscal\cdot \cscal^{-1} (\cscal^{1/2})^2 \partial_\curubar \metg_{\sca\scb} = \cscal\cdot\chihati_{\sca\scb} = \Osize(\sscal).\label{chihatsicondb}
\end{equation}
Now on $\init_1$, $\tr \chisi$ must satisfy the Raychaudhuri equation (\cite{hawkingEllis}, \cite{Poisson})
\begin{equation}\label{raychaudhuria}
\partial_{\curubars} \tr \chisi = -\frac{1}{2} (\tr \chisi)^2 - |\chihatsi|^2,
\end{equation}
while the condition that the spacetime be Minkowskian for $\curubars \leq 0$ gives the initial condition
\begin{equation}\label{raychaudhuriainit}
\tr \chisi|_{\curubars = 0} = \frac{2\cscal}{\currz}.
\end{equation}
From this we get two results. First, applying (\ref{chihatsicondb}) and Lemma \ref{OsizeODE}, we find that
\begin{equation}
\tr\chisi = \Osize(\sscal) + \Osize(\cscal).\label{trchisiOsizea}
\end{equation}
Second, integrating $\partial_{\curubars} \tr\chisi \leq -\frac{1}{2} (\tr\chisi)^2$, we obtain the simple $L^\infty$ bound
\begin{equation}
\tr \chisi \leq \frac{2\cscal}{\currz + \cscal \curubars} = \Osize(\cscal);\label{trchisilbounda}
\end{equation}
substituting this back into (\ref{raychaudhuria}) and integrating gives on $\init_1$ the $L^\infty$ bound
\begin{equation}\label{trchisibounds}
\tr \chisi \geq \frac{2\cscal}{\currz + \cscal \curubars} - \int_0^\curubars |\chihatsi|^2\,d\curubars' = \Osize(\cscal) + \Osize(\sscal^2).
\end{equation}
We emphasize that (\ref{trchisilbounda}) -- (\ref{trchisibounds}) do {\it not\/} directly imply (\ref{trchisiOsizea}). We shall use (\ref{trchisibounds}) below to show that $\init_1$ does not contain trapped surfaces.
\par
Now write $\metg_{\sca\scb}|_{\init_1} = \conffact \initmet_{\sca\scb}$ for $\det\initmet = 1$. Then it is not hard to show that
\begin{equation}
\partial_\curubars \log \conffact = \tr \chisi
\end{equation}
and
\begin{equation}
\partial_\curubars \initmet = \frac{2}{\conffact} \chihatsi.
\end{equation}
We may thus solve for $\conffact$ and then $\initmet$. The bounds above clearly imply that $\initmet = \Osize(\sscal)$. For data on $\init_0$ we set the data for the $\cscal$-scaled Minkowski metric as given in (\ref{modscaletap}); for the data on the intersection sphere, $\conffact_1$, $\conffact_2$, $\metd_1$, and $\metd_2$, we likewise use the data arising from (\ref{modscaletap}). By Theorem \ref{rendallexistence}, then, there is a neighborhood of $(0, 0, 0, 0)$, say $\solvreg = [0, \curusff) \times [0, \curubarsff) \times U$, on which the Einstein vacuum equations possess a solution $\metg$ realizing the given initial data. Moreover, by the smoothness property in Theorem \ref{rendallexistence}, it is clear that $\metg = \Osize(\cscal) + \Osize(\sscal)$, and that the region $\solvreg$ can be taken to be independent of $\cscal$, $\sscal$, $\dscal$. As noted above, this completes the proof of the existence portion of the above theorem.
\par
Applying standard existence and smoothness results from ODE theory (see, e.g., \cite{sideris}, Chapter 6), together with the precompactness of $\solvreg$, it is then not hard to see that we may cover a (possibly smaller) neighborhood of $\solvreg$ with geodesics starting from $\{ \curubar = 0 \}$ and initially parallel to $\L$. This will give a map $\geodmap : [0, \curmaxLaff) \times [0, \curusff) \times U \rightarrow \solvreg$ which covers some neighborhood of $(0, 0, 0, 0)$. To see that $\geodmap$ is injective, note that at $\cscal = \sscal = 0$ $\geodmap$ is simply the identity, so that from $\metg = \Osize(\cscal) + \Osize(\sscal)$ and smoothness considerations it is clear that the derivative of $\geodmap$ satisfies $D\geodmap - \id = \Osize(\cscal) + \Osize(\sscal)$. Lemma 6.1.2 in \cite{spntcarXiv} then shows that $\geodmap$ must be injective, from which we see that it is a diffeomorphism onto its image. Thus by shrinking $\solvreg$ if necessary we may assume that it is covered by a null geodesic coordinate system along $\L$. Moreover, the maximum affine parameter value $\curmaxLaff$ can be taken independent of $\cscal$, $\sscal$, $\dscal$, and (using compactness of $\Sph^2$ again) $\curangulvar$.
We shall define a spatial frame by parallely-transporting the frame $\eframe_1, \eframe_2$ in (\ref{etaibframedef}) along $\Ls$, see (\ref{pteframedef}) below.
\par
Since all quantities in Lemma \ref{etaibprops} are smooth in the metric, while clearly $\metg = \etaib$ at $\cscal = \sscal = 0$, their values with respect to $\metg$ will differ from their values in Lemma \ref{etaibprops} by terms $\Osize(\cscal) + \Osize(\sscal)$ by Proposition \ref{Osizechar}. This is a key point which will be crucial in proving the existence of trapped surfaces, to which we now turn.
\par
Now suppose that $\chihati$ satisfies (\ref{angdiffbounds})-(\ref{shearbounds}), and that the condition $\cscal = \cssscoeff \sscal^2$ is satisfied for some constant $\cssscoeff$.
We will now show that the solution just constructed possesses a closed trapped surface. We will use a mechanism similar to that in Christodoulou \cite{christodoulou} and Klainerman and Rodnianski \cite{klainrodi}, though we will derive all needed equations anew in our setting, following Christodoulou's treatment where relevant.
\par
(As an aside, for readers familiar with those papers, we give the following informal overview. Consider the following null structure and constraint equations (see, e.g., \cite{klainrodi}, (47), (48), (49), (51))
\begin{align}
\nabla_3 \hat{\chi} + \frac{1}{2} \tr\underline{\chi} \hat{\chi} &= \nabla \hat{\otimes} \eta + 2\underline{\omega} \hat{\chi} - \frac{1}{2} \tr\chi \hat{\underline{\chi}} + \eta \hat{\otimes} \eta,\label{chihatevoleq}\\
\nabla_4\eta &= -\chi \cdot (\eta - \underline{\eta}) - \beta,\\
\text{div}\,\hat{\chi} &= \frac{1}{2} \nabla\tr\chi - \frac{1}{2} (\eta - \underline{\eta}) \cdot (\hat{\chi} - \frac{1}{2} \tr\chi) - \beta,\\
\end{align}
where we note that for us $\nabla_4 \sim \nabla_\L$, $\nabla_3 \sim \nabla_\Lbar$ (in the unscaled coordinates), and refer to \cite{klainrodi} for definitions of the quantities and notation used here. Now it is not hard to see that, up to terms of size $\Osize(\sscal^2)$, the equations above should reduce to
\begin{align}
\nabla_3 \hat{\chi} + \frac{1}{2} \tr\underline{\chi} \hat{\chi} &= \nabla \hat{\otimes} \eta,\\
\beta &= -\nabla_4\eta,\\
\beta &= \text{div}\,\hat{\chi}.
\end{align}
Differentiating the first of these using $\nabla_4$, ignoring terms resulting from commuting derivatives, and using the final two equations, we obtain an equation of the form
\begin{equation}
\nabla_4 \nabla_3 \hat{\chi} + \frac{1}{2} \tr\underline{\chi} \hat{\chi} = -\nabla\hat{\otimes}\text{div}\,\hat{\chi}.
\end{equation}
This is very close to a flat-space (Minkowski) wave equation for (some multiple of) $\hat{\chi}$. The next couple pages will make all of the above manipulations rigorous in our setting.)
\par
Define $\Lbars = \partial_\curus$, $\Ls = \partial_\curubars$. We will work in the frame $\{ \Lbars, \Ls, \eframe_1, \eframe_2 \}$ obtained by parallely-transporting the vectors
\begin{equation}
\eframei_1 = \frac{1}{\currz - 2^{-1/2} \curus} \partial_\thetas,\quad \eframei_2 = \frac{1}{(\currz - 2^{-1/2} \curus) \cos(\cscal^{1/2} \thetas)} \partial_\phis\label{pteframedef}
\end{equation}
along $\Ls$. We use indices $\framea$, $\frameb$, $\framec$, etc., to denote indices in the frame $\{ \eframe_1, \eframe_2 \}$. In this frame, we have
\begin{equation}
\chi_{\framea\frameb} = \metg(\nabla_{\eframe_\framea} \Ls, \eframe_\frameb).
\end{equation}
We begin by seeking an equation analogous to (\ref{chihatevoleq}). We shall call a term {\it ignorable\/} if it is of size $\Osize(\sscal^\ell)$ where $k \geq 2$; recall that we have set $\cscal = \sscal^2$. We calculate:
\begin{equation}
\nabla_{\Lbars} \chi_{\framea\frameb} = \metg(\nabla_\Lbars \nabla_{\eframe_\framea} \Ls, \eframe_\frameb) + \metg(\nabla_{\eframe_\framea} \Ls, \nabla_{\Lbars} \eframe_\frameb).
\end{equation}
Now by Lemma \ref{etaibprops}, we have
\begin{equation}
\nabla \Ls = \Osize(\cscal),\,\nabla_{\Lbars} \eframe_\framea = \Osize(\cscal),
\end{equation}
so that the last term above is ignorable. Commuting $\nabla_\Lbars$ and $\nabla_{\eframe_\framea}$, we obtain
\begin{equation}\label{derivchievoleqa}
\nabla_{\Lbars} \chi_{\framea\frameb} = \metg(\nabla_{\eframe_\framea} \nabla_\Lbars \Ls, \eframe_\frameb) + \Riem(\eframe_\frameb, \Ls, \Lbars, \eframe_\framea) + \metg(\nabla_{[\Lbars, \eframe_\framea]} \Ls, \eframe_\frameb) + \Osize(\sscal^2).
\end{equation}
From Lemma \ref{etaibprops} again, we obtain
\begin{equation}
\metg(\nabla_{[\Lbars, \eframe_\framea]} \Ls, \eframe_\frameb) = \frac{1}{\sqrt{2}(\currz - 2^{-1/2} \curus)} \chi_{\framea\frameb} + \Osize(\sscal^2).
\end{equation}
Further, noting that in the frame $\{ \Ls, \Lbars, \eframe_1, \eframe_2 \}$ the inverse metric is $\Osize(\sscal)$ away from that in Minkowski space, we find that
\begin{multline}
\Ric(\eframe_\framea, \eframe_\frameb) = -\Riem(\Ls, \eframe_\framea, \Lbars, \eframe_\frameb) - \Riem(\Lbars, \eframe_\framea, \Ls, \eframe_\frameb)\\
{}+ \sum_\framec \Riem(\eframe_\framec, \eframe_\framea, \eframe_\framec, \eframe_\frameb) + \Osize(\sscal^2);
\end{multline}
setting $\framea = 1$, $\frameb = 2$, we obtain
\begin{equation}
\Riem(\Ls, \eframe_1, \Lbars, \eframe_2) + \Riem(\Ls, \eframe_2, \Lbars, \eframe_1) = -\Ric(\eframe_1, \eframe_2) + \Osize(\sscal^2),
\end{equation}
while setting $\framea = \frameb$ we obtain
\begin{equation}
\Riem(\Ls, \eframe_\framea, \Lbars, \eframe_\framea) = -\frac{1}{2} \Ric(\eframe_\framea, \eframe_\framea) + \frac{1}{2} \Riem(\eframe_1, \eframe_2, \eframe_1, \eframe_2) + \Osize(\sscal^2).
\end{equation}
Since $\Ric = 0$ for a solution to the Einstein vacuum equations, we deduce that the term $\Riem(\eframe_\frameb, \Ls, \Lbars, \eframe_\framea)$ in (\ref{derivchievoleqa}) must be a pure trace, up to terms of size $\Osize(\sscal^2)$.
\par
It thus remains only to treat the term $\metg(\nabla_{\eframe_\framea} \nabla_\Lbars \Ls, \eframe_\frameb)$ in (\ref{derivchievoleqa}). (This term is analogous to the one non-ignorable term $\nabla \hat{\otimes} \eta$ in the evolution equation for $\hat{\chi}$ in Christodoulou's setting as described above.) We observe that we may write
\begin{equation}
\metg(\nabla_{\eframe_\framea} \nabla_\Lbars \Ls, \eframe_\frameb) = -\metg(\nabla_\Lbars \Ls, \nabla_{\eframe_\framea} \eframe_\frameb) + \nabla_{\eframe_\framea} \metg(\nabla_{\Lbars} \Ls, \eframe_\frameb);
\end{equation}
now by Lemma \ref{etaibprops} again, $\nabla_{\eframe_\framea} \eframe_\frameb$ is a pure trace plus $\Osize(\sscal)$, while $\nabla \Ls = \Osize(\sscal)$, so that the first term is a pure trace plus a term of size $\Osize(\sscal^2)$. The second term is slightly more tricky (in Christodoulou's setting, we need to introduce the equations for $\nabla_4 \eta$ and $\text{div}\,\chi$). We observe
\begin{align}\label{nablaLsmetgeq}
\nabla_\Ls \metg(\nabla_\Lbars \Ls, \eframe_\frameb) &= \metg(\nabla_\Ls \nabla_\Lbars \Ls, \eframe_\frameb) + \metg(\nabla_\Lbars \Ls, \nabla_\Ls \eframe_\frameb)\\
&= \Riem(\eframe_\frameb, \Ls, \Ls, \Lbars) + \metg(\nabla_{[\Ls, \Lbars]} \Ls, \eframe_\frameb),
\end{align}
where we have used $\nabla_\Ls \Ls = 0$, $\nabla_\Ls \eframe_\framea = 0$. Since $[\Ls, \Lbars] = \Osize(\sscal)$, the second term above is of size $\Osize(\sscal^2)$ and hence ignorable. Now as before we have
\begin{multline}
\Ric(\eframe_\frameb, \Ls) = -\Riem(\Lbars, \eframe_\frameb, \Ls, \Ls) - \Riem(\Ls, \eframe_\frameb, \Lbars, \Ls)\\
{}+ \sum_\framec \Riem(\eframe_\framec, \eframe_\frameb, \eframe_\framec, \Ls) + \Osize(\sscal^2),
\end{multline}
so that for $\Ric(\eframe_\frameb, \Ls) = 0$ we have
\begin{equation}\label{Riemtraca}
\Riem(\Ls, \eframe_\frameb, \Lbars, \Ls) = \sum_\framec \Riem(\eframe_\framec, \eframe_\frameb, \eframe_\framec, \Ls) + \Osize(\sscal^2).
\end{equation}
Further, 
\begin{align}
\nabla_{\eframe_\framea} \chi_{\frameb\framec} - \nabla_{\eframe_\framec} \chi_{\frameb\framea} &= \nabla_{\eframe_\framea} \metg(\nabla_{\eframe_\framec} \Ls, \eframe_\frameb) - \nabla_{\eframe_\framec} \metg(\nabla_{\eframe_\framea} \Ls, \eframe_\frameb)\\
&= \Riem(\eframe_\frameb, \Ls, \eframe_\framea, \eframe_\framec) + \metg(\nabla_{[\eframe_\framea, \eframe_\framec]} \Ls, \eframe_\frameb) + \Osize(\sscal^2),
\end{align}
where we have used the fact that (see Lemma \ref{etaibprops})
\begin{align}
\metg(\nabla_{\eframe_\framec} \Ls, \nabla_{\eframe_\framea} \eframe_\frameb) &= \Osize(\sscal^2) + \metg(\nabla_{\eframe_\framec} \Ls, -\frac{1}{\sqrt{2} (\currz - 2^{-1/2} \curus)} \delta_{\framea\frameb} \partial_\curubars)\\
&= \Osize(\sscal^2) - \frac{1}{\sqrt{2} (\currz - 2^{-1/2} \curus)} \metg(\nabla_{\eframe_\framec} \Ls, \Ls) = \Osize(\sscal^2).
\end{align}
Similarly, from Lemma \ref{etaibprops} it is easy to see that $\metg(\nabla_{[\eframe_\framea, \eframe_\framec]} \Ls, \eframe_\frameb) = \Osize(\sscal^2)$. Tracing then gives
\begin{align}
\metg^{\framea\frameb} \nabla_{\eframe_\framea} \chi_{\frameb\framec} - \nabla_{\eframe_\framec} \tr \chi &= \Riem(\Ls, \eframe_\framec, \Lbars, \Ls) + \Osize(\sscal^2)\\
&= \nabla_\Ls \metg(\nabla_{\Lbars} \Ls, \eframe_\framec) + \Osize(\sscal^2)
\end{align}
by (\ref{Riemtraca}) and (\ref{nablaLsmetgeq}). We thus obtain (noting that $\metg(\nabla_\Lbars \Ls, \eframe_\frameb)$ is a scalar, and recalling that $\tr \chi = \Osize(\sscal^2)$)
\begin{align}
\nabla_\Ls \metg(\nabla_{\eframe_\framea} \nabla_\Lbars \Ls, \eframe_\frameb) &= \Osize(\sscal^2) + \nabla_{\eframe_\framea} \nabla_\Ls \metg(\nabla_{\Lbars} \Ls, \eframe_\frameb)\\
&= \Osize(\sscal^2) + \nabla_{\eframe_\framea} \left( \metg^{\framec\framed} \nabla_{\eframe_\framec} \chi_{\frameb\framed}\right),
\end{align}
and finally
\begin{equation}
\nabla_\Ls\nabla_{\Lbars} \chi_{\framea\frameb} = \frac{1}{\sqrt{2}(\currz - 2^{-1/2} \curus)} \chi_{\framea\frameb} + \nabla_{\eframe_\framea} \metg^{\framec\framed} \nabla_{\eframe_\framec} \chihats_{\frameb\framed} + \Osize(\sscal^2).\label{chihatwavesemifinal}
\end{equation}
Now
\begin{align}
\nabla_{\eframe_\framea} \metg^{\eframec\eframed} \nabla_{\eframe_\framec} \chihats_{\eframeb\eframed} &= \nabla_{\eframe_\framea} \nabla_{\eframe_1} \chihats_{\eframeb 1} + \nabla_{\eframe_\framea} \nabla_{\eframe_2} \chihats_{\eframeb 2} \\
&= \left( \begin{matrix} \nabla_{\eframe_1}^2 \chihats_{11} + \nabla_{\eframe_1} \nabla_\eframe{2} \chihats_{12} & \nabla_{\eframe_1}^2 \chihats_{21} + \nabla_{\eframe_1} \nabla_{\eframe_2} \chihats_{22} \\
\nabla_{\eframe_2} \nabla_{\eframe_1} \chihats_{11} + \nabla_{\eframe_2}^2 \chihats_{12} & \nabla_{\eframe_2} \nabla_{\eframe_1} \chihats_{21} + \nabla_{\eframe_2}^2 \chihats_{22} \end{matrix} \right)\label{spderivchihat}
\end{align}
the trace-free symmetric part of which consists of the two terms (using $\chihats_{11} = -\chihats_{22}$, $\chihats_{12} = \chihats_{21}$) 
\begin{gather}
\frac{1}{2} (\nabla_{\eframe_1}^2 \chihats_{11} + \nabla_{\eframe_2}^2 \chihats_{22} + [\nabla_{\eframe_1}, \nabla_{\eframe_2} ]\chihats_{12} ),\\
\frac{1}{2} ( \nabla_{\eframe_1}^2 \chihats_{12} + \nabla_{\eframe_2}^2 \chihats_{12} + [\nabla_{\eframe_1}, \nabla_{\eframe_2}] \chihats_{22} ),
\end{gather}
from which it is not hard to see (using the fact that on scalars $\nabla_{\eframe_1}$ and $\nabla_{\eframe_2}$ commute up to terms of size $\Osize(\sscal)$) that the trace-free symmetric part of (\ref{spderivchihat}) is simply
\begin{equation}
\frac{1}{2} (\nabla_{\eframe_1}^2 + \nabla_{\eframe_2}^2) \chihats_{\framea\frameb} + \Osize(\sscal^2).
\end{equation}
Substituting this back into (\ref{chihatwavesemifinal}), replacing derivatives along $\Ls$ and $\Lbars$ by $\partial_\curubars$ and $\partial_\curus$ (which, since we are applying them to $\chihats$, will result in an error of size $\Osize(\sscal^2)$) and simplifying, we obtain finally that
\begin{multline}
-2 \partial_{\curus} \partial_\curubars [\sqrt{2}(\currz - 2^{-1/2}\curus) \chihats_{\framea\frameb}]\\
+ \frac{1}{(\currz - 2^{-1/2}\curus)^2} [\partial_\thetas^2 + \partial_\phis^2] [\sqrt{2}(\currz - 2^{-1/2} \curus) \chihats_{\framea\frameb}] = \Osize(\sscal^2).
\end{multline}
Now define on $\solvreg$
\begin{equation}
\chihatgo = \frac{\currz}{\currz - \curus} \chihats|_{\curus = 0}.
\end{equation}
Then on $\solvreg$
\begin{multline}
\{ -2\partial_\curus \partial_\curubars + \frac{1}{(\currz - 2^{-1/2}\curus)^2} [\partial_\thetas^2 + \partial_\phis^2]\} [(\currz - \curus) \chihatgo]\\
= \frac{\currz}{(\currz - 2^{-1/2} \curus)^2} [\partial_\thetas^2 + \partial_\phis^2] \chihatsi = \Osize(\dscal\sscal)
\end{multline}
by assumption. If we now assign $\sscal = \sscal(\dscal)$, $\sscal(\dscal)$ smooth, $\sscal(0) = 0$, then $\Osize(\sscal^2) + \Osize(\dscal\sscal) = \Osize(\dscal\sscal)$ and we may write
\begin{equation}
\{ -2\partial_\curus \partial_\curubars + \frac{1}{(\currz - 2^{-1/2} \curus)^2} [\partial_\thetas^2 + \partial_\phis^2]\} [(\currz - \curus) (\chihats - \chihatgo)] = \Osize(\dscal\sscal).
\end{equation}
Now we note that Rendall's argument concerning smoothness of solutions to partial differential equations with initial data depending smoothly on a parameter can readily be extended to equations of the above form where there is a forcing term depending smoothly on a parameter; thus we obtain
\begin{equation}
(\currz - \curus) (\chihats - \chihatgo) = \Osize(\dscal\sscal),
\end{equation}
and finally
\begin{equation}
\chihats_{\framea\frameb} = \frac{\currz}{\currz - \curus} \chihatsi_{\framea\frameb} + \Osize(\dscal\sscal).
\end{equation}
We now wish to integrate the Raychaudhuri equation along a geodesic starting from $\{ \curubars = 0 \}$ and passing through $\geodmap(\Lafftr, \curustr, \thetas, \phis) = (\curustr, \curubarstr, \thetas, \phis) + \Osize(\sscal)$, for some specific values of $\Lafftr$, $\curustr$. We note that such a geodesic, together with its tangent vector, will differ from the coordinate line $(\curus, \thetas, \phis) = \text{constant}$ by terms of size $\Osize(\sscal)$. Thus when integrating $|\chihats|^2$ in the Raychaudhuri equation we may replace the geodesic by this coordinate line and absorb the $\Osize(\sscal^3)$ error terms into $\Osize(\dscal\sscal^2)$, obtaining
\begin{align}
\tr\chi &\leq \frac{2\cscal}{\currz - \curustr} - 2\left(\frac{\currz}{\currz - \curustr}\right)^2 \int_0^{\curubars} \left\{|\chihats_{11}|^2 + |\chihats_{12}|^2\right\}|_{\curus = 0} + \Osize(\dscal\sscal^2)\\
&= \Osize(\dscal\sscal^2) + \frac{2\cscal}{\currz - \curustr} \left[ 1 - \frac{\currz^2 \int_0^{\curubars} |\chihats_{11}|^2 + |\chihats_{12}|^2}{\currz - \curustr}\right].
\end{align}
If the final lower bound on $\int_0^{\curubars} |\chihats_{11}|^2 + |\chihats_{12}|^2$ in (\ref{shearbounds}) is satisfied, then we will have $\tr\chi < 0$ at $\geodmap(\Lafftr, \curustr, \thetas, \phis)$. Note that we may choose $\curustr$ and $\Lafftr$ independently of $\thetas$, $\phis$. 
If necessary, we may now adjust the frame $\{ \eframe_\framea \}$ by adding a multiple of $\L$ so as to make it tangent to the surface $\{ \curus = \curustr \} \cap \{ \Laff = \Lafftr \}$, without affecting its orthonormality or $\tr\chi$.
\par
To sum up, we have shown that for every $\curangulvar \in \Sph^2$ there are $\curustr$ and $\Lafftr$ such that a patch of surface in $\{ \curus = \curustr \} \cap \{ \Laff = \Lafftr \}$ around $\{ \curangulvar = (\thetas, \phis) \}$ is trapped. Since the coordinates $\curus$, $\curubars$, as well as the values of $\Laff$ along null geodesics, agree across different coordinate patches in $\Sph^2$, these different patches may be combined into a smooth closed surface which is moreover trapped. This completes the proof. 
\end{proof}
\section{Discussion}\label{discussion}
Perhaps one of the single largest issues which arises in attempting to apply the current method is the necessity of a {\it smooth}, {\it nondegenerate\/} limit metric as $\cscal \rightarrow 0^+$. The limit metric in the unscaled coordinates is highly singular and restricted to a codimension-1 hypersurface, necessitating the unusual step we have taken here of performing a parameter-dependent coordinate change.
\par
Three potential future applications of this method are as follows. Since the existence portion of the above proof is only weakly dependent on the geometry of the intersection 2-surface, we believe that the above results can be adapted to the case where this intersection 2-surface is a spheroid rather than a sphere. Such results have not appeared in the literature, to our knowledge, though it seems possible that the results in \cite{klainchen} could be used to prove a trapped surface formation result for a spheroid with very small (on the order of some positive power of $\delta$) eccentricity. Second, we have indicated above precisely where we have used the Einstein {\it vacuum\/} equations in the foregoing, and it would be of interest to extend the results above to the case where a matter field is also present. Third, we believe that the techniques above should shed light on the recent censorship results of An \cite{ancensor}.
\bibliographystyle{amsplain}
\bibliography{sphere_scale}
\end{document}